\documentclass[conference,a4paper]{IEEEtran}
\usepackage{bbding}
\usepackage[T1]{fontenc}
\usepackage{aurical}
\usepackage{huncial}
\usepackage{linearb}
\usepackage{textcomp}
\usepackage{amssymb,epic,eepic}
\usepackage{amsmath}
\usepackage[amsmath,amsthm,thmmarks]{ntheorem}
\usepackage{amsfonts}
\usepackage{stmaryrd}
\usepackage{capt-of}
\usepackage{graphicx}
\usepackage[usenames,dvipsnames]{pstricks}
\usepackage{epsfig}
\usepackage{pst-grad}
\usepackage{pst-plot}
\usepackage{mathtools}
\usepackage{textfit} 
\usepackage{subfigure}
\usepackage{float}

\usepackage{comment}

\newtheorem{theorem}{Theorem} 

\newtheorem{lemma}[theorem]{Lemma}

\newtheorem{conjecture}[theorem]{Conjecture}

\begin{document}

\title{Is Winter Coming?}

\author{\IEEEauthorblockN{A. Winter}
\IEEEauthorblockA{\textit{Grup d'Informaci\'o Qu\`antica}\\
\textit{Departament de F\'isica}\\
\textit{Universitat Aut\`onoma}\\
\textit{de Barcelona, Spain}\\
andreas.winter@uab.cat} 

\and
\IEEEauthorblockN{A. Winter}
\IEEEauthorblockA{\textit{ICREA--Instituci\'{o}}\\
\textit{Catalana de Recerca}\\
\textit{i Estudis Avan\c{c}ats}\\
\textit{Barcelona, Spain}}

\and
\IEEEauthorblockN{A. Winter}
\IEEEauthorblockA{\textit{Institute for}\\
\textit{Advanced Study}\\
\textit{Technische Universit\"at}\\
\textit{M\"unchen, Germany}} 

\and
\IEEEauthorblockN{A. Winter}
\IEEEauthorblockA{\textit{QUIRCK--Quantum}\\
\textit{Independent Research}\\ 
\textit{Center Kessenich}\\
\textit{Bonn, Germany}}
}

\maketitle


\begin{abstract}
We critically examine the often-made observation that 
``quantum winter [or some other winter] is coming'', and the related 
admonition to prepare for this or that winter, inevitably bound to arrive. 
What we find based on even the most superficial look at the available 
evidence is that such statements not only are overblown hype, but are also 
factually wrong: Winter is here, and the real question is rather for how 
long it/they will stay.
\end{abstract}


\medskip\noindent
\textit{\textbf{1. Introduction.}}
It has been suggested, and is now being repeated often without proper
attribution, that ``quantum winter is coming''. Based on the familiar 
seasonal cycle, such a statement might be considered so banal as to be 
almost empty, even if referring metaphorically to the so-called 
``hype cycle'' of fashionable and overly-promoted fields of research. 
Quantum winter in this sense is a source of considerable anxiety 
among members of the currently objectively overhyped field of quantum 
science and technology. The field of artificial intelligence has gone 
through several cycles of hype and disillusionment, and indeed AI winter, 
the original ``winter'' of what is now recognised as a collection of similar 
phenomena across different fields of scientific research and/or technological 
innovation, has been proclaimed repeatedly since the 1980s \cite{AW:AI}. 

Rather obviously, a lot of this talk is merely the product of minds 
intoxicated by the rush of attention (both socially and monetary), 
and at the same time dreading heavily the end of the halcyon days of easy 
living off the hype. 
Quite characteristically, there is usually a secondary investment winter
affecting those financing the hype and rightly fearing for the expected 
return on their investment. 
The whole thing is not unlike the gold rushes of another age, when 
desperate, crazed diggers and ruthless entrepreneurs turned entire 
regions upside down, only to move on when the fever had passed and 
the ground did not yield any more of the precious metal \cite{AW:goldrush}. 
The market pressure exerted by modern-day investors sometimes, paradoxically, 
prolongs the hype phase, preventing the drop into the wintry 
valley of disillusionment, as derivatives are created (spin-off 
companies floated on the stock market, etc) whose speculation value 
can carry the bubble for some time when its practitioners are already 
preparing for the crash. 
It is precisely this -- partial or repressed -- knowledge of a bubble ripe 
for bursting that creates the broadly shared dread of the ``winter'' to come. 
Note that economically driven hypes are nothing recent in Western economies, 
as evidenced by the famous tulip mania in the Dutch 17th century \cite{AW:tulip}
(cf.~\cite{AW-1}). 

While we are personally sympathetic to the psychological conditions affecting 
our dear colleagues, as scientists it is our duty to regard the matter with 
the cold eye of empiricism. As we will explain in the next two 
sections, the fear of any winter coming is entirely unfounded, 
principally because winter, in all its forms and manifestations, 
has long arrived. What conclusions can be drawn from this fact 
is discussed in the final section.

\medskip\noindent
\textit{\textbf{2. Methods.}}
We follow the model of \cite{AW0}, elaborated in \cite{AW1,AW2,AW3,AW4} and subsequent 
investigations, to estimate the proximity of winter. Our method is 
scientific and self-referential, i.e.~bibliographic. 
We need a small auxiliary result. 

\begin{lemma}[cf.~A.~Winter~\cite{AW5}]
\label{lemma:Winter}
There exists an integer $n$, significantly smaller than one might believe, 
such that the reader of a scientific article, when confronted with $n$ 
or more arbitrary references, will accept anything the paper 
claims to follow from them.
\end{lemma}
\begin{proof}\textit{(Sketch)}
To be added in the journal version, see however \cite{AW6,AW7} and 
\cite{AW7.1,AW7.2} for similar statements. 
\end{proof}

\medskip\noindent
\textit{\textbf{3. Winter's universality.}}
With these surprisingly elementary preparations we are ready for 
our main result. It is enough to regard 
Refs.~\cite{AW8,AW9,AW10,AW11,AW12,AW13,AW14,AW15}, 
in particular \cite{AW16}, followed by the litany of 
\cite{AW17,AW18,AW19,AW20,AW21,AW22,AW23,AW24,AW25,AW26,AW27,AW28,AW29}, 
and to consider that this is only a small and sufficiently random sample 
of wintry \emph{selecta}, to realise beyond doubt that Winter has been 
around quite forcefully for decades, if not centuries. 

This clearly means that any fear of Winter coming is completely 
irrational: there is nothing to be prepared about, to ``winterise'' 
as it were, as we are right in the middle of it. 
We allow ourselves the expression of the modest hope that this 
insight will lead to the acceptance of the universal presence of Winter 
in all walks of (academic) life \cite{AW:justiz}.

\medskip\noindent
\textit{\textbf{4. Epilogue.}}
We regard our arguments as irrefutable. Hence, as a next step, we need 
to prepare for the phase after acceptance: the hype. That there are 
already indications of a wintering crescendo is evident from the 
literature \cite{AW30,AW31,AW32}.
So, will a Winter winter come, and should we worry about it? 

More to the point, shall we compare something to a summer's day? 
Of course, one swallow does not a summer make. 
However, inspired by these expressions of somewhat complementary sentiments, 
and taking the theme to its natural conclusion, 
we are motivated to make the following conjecture. 

\begin{conjecture}
\label{conj:winter}
Now is the Winter of our discontent made glorious summer by this sun of York. 
\end{conjecture}

This might appear at first sight rather bold, and will perhaps be 
suspected of being indebted more to wishful thinking than having 
any basis in solid evidence. We invite the interested reader to 
try and falsify it, as we have tried -- and failed. In the appendix, 
A. Winter graciously shares her extensive remarks on possible approaches 
towards this conjecture, and on interesting consequences of it.

\medskip\noindent
\textit{\textbf{Acknowledgments.}}
The authors thank A. Winter and A. Winter for their thoughtful and 
interesting comments, in particular with regard to A. Winter's contribution
to the field. 
They are grateful to A. Winter for spotting an error in an early version 
of Lemma \ref{lemma:Winter}. Finally, they acknowledge A. Winter's 
thorough proofreading. 
The first (AW) and fourth (AW) author were supported by the 
Andrew Andreevich and Andrea Winter Foundation through a generous 
research stipend. 
A. Winter declares no competing financial interests.

 
\medskip\noindent
\textit{\textbf{Appendix --- Remarks on Conjecture \ref{conj:winter} (by A. Winter).}}
This was to be something profound reflecting my complex and nuanced
opinions about the provocative work of A. Winter \emph{et al.}, 
but the present paragraph is too short to contain it. 
For

\end{document}